\documentclass[a4paper]{article}
\usepackage{pbox}
\usepackage{url}
\usepackage{hyperref}
\usepackage{array}
\usepackage{tabularx}
\usepackage{graphicx,float}
\usepackage[left]{lineno}
\usepackage{amsmath,amssymb}
\usepackage{amsthm}
\usepackage{authblk} 
 \newtheorem{theorem}{Theorem}

\newtheorem{definition}{Definition}

\title{Simple, compact and robust approximate string dictionary\thanks{
Most of this work was done when the second author was at LIAFA, University Paris Diderot - Paris 7. Part of the work was done while the second author was visiting LIAFA. The work was partially supported by the french
 ANR project MAPPI (project number ANR-2010-COSI-004).}}
\author[1]{Ibrahim Chegrane}
\author[2]{Djamal Belazzougui}
\affil[1]{USTHB - Faculty of Electronics and Computer Science, Laboratory of Artificial Intelligence (LRIA), Algeria}
\affil[2]{Helsinki Institute for Information Technology (\textsc{HIIT}),
Department of Computer Science, University of Helsinki}

\begin{document}
\maketitle
\begin{abstract}
This paper is concerned with practical implementations of approximate string dictionaries that allow edit errors. In this problem, we have as input a dictionary $D$ of $d$ strings of total length $n$ over an alphabet of size $\sigma$. Given a bound $k$ and a  pattern $x$ of length $m$, a query has to return all the strings of the dictionary which are at edit distance at most $k$ from $x$, where the edit distance between two strings $x$ and $y$ is defined as the minimum-cost sequence of edit operations that transform $x$ into $y$. The cost of a sequence of operations is defined as the sum of the costs of the operations involved in the sequence. 
In this paper, we assume that each of these operations has unit cost and consider only three operations: deletion of one character, insertion of one character and substitution of a character by another. 
We present a practical implementation of the data structure we recently proposed and which works only for one error. We extend the scheme to $2\leq k<m$. 

Our implementation has many desirable properties: it has a very fast and space-efficient building algorithm. The dictionary data structure is compact and has fast and robust query time. 

Finally our data structure is simple to implement as it only uses basic techniques from the literature, mainly 
hashing (linear probing and hash signatures) and succinct data structures (bitvectors supporting rank queries). 
\end{abstract}
\maketitle
%\end{frontmatter}

\section{Introduction}
The problem of approximate dictionaries is widely studied. The problem is defined as follows: we are given (in advance) a dictionary $D=\{x_1,x_2,\ldots x_d\}$ of $d$ words (strings) of total length $n$ over an alphabet of size $\sigma$~\footnote{For practical purpose we assume the alphabet to be the integer domain $[1..\sigma]$.} and a threshold $k$. We want to build a data structure on $D$ so as to be able to answer the following queries: given a string $x$ of length $m$, return all the words of the dictionary which are 
at edit distance at most $k$ from $x$, where the edit distance 
(or Levenshtein distance~\cite{Lev66}) between two strings 
is defined as the minimum-cost sequence of edit operations needed to transform one of the strings into the other.
The cost of a sequence of edit operations is simply the sum of the costs of individual operations. 

In this paper, we assume that each of these operations has unit cost,
and only consider three kinds of edit operations: insertions, deletions and substitutions. 

The important quantitative parameters to judge the quality of solutions to this problem are :
\begin{enumerate}
\item The space usage of the solution in number of bits. An ideal solution should use $n\log\sigma$ bits (or less if some kind of compression is used). 
\item The time needed to build the dictionary data structure. 
\item The time needed to answer a query. 
\end{enumerate}

In addition to the three parameters above, other important considerations are whether the solution is incremental or not (whether it supports fast addition of strings to the dictionary) and the simplicity of the solution (whether the solution uses sophisticated data structures or not).

The problem has many applications. Among the most important ones, we mention: 
\begin{itemize}
\item Spell checking  in text editors~\cite{pollock1984automatic}. 
\item Correction of user typing errors in search engines~\cite{cucerzan2004spelling}.  
\item Post-processing for optical character recognition~\cite{reynaert2008non}. 
\item Online data cleaning in databases~\cite{chaudhuri2003robust}. 
\item Password cracking~\cite{manber1994algorithm}. 
\end{itemize}
Many solutions to the approximate dictionary problem under edit distance have been proposed in the literature~\cite{Na01,Bo11}. 

Some of those solutions are based on heuristics which can have good average-case performance on the assumption that the dictionary strings come from some random distribution. 
However those solutions usually come with a bad worst-case performance. In fact it is easy to come up with a worst-case 
example for many of those solutions. As an example we take the following heuristic which is probably the most popular one: divide the query string $x$ into $k+1$ pieces, then we are sure that one of the pieces will appear as a substring of one of the words in the dictionary if that word happens to be at edit distance at most $k$ from $x$. 

Thus a query can restrict its attention only to the dictionary strings that contain one of the pieces of the query string. 

The heuristic above has been discovered and rediscovered many times (even from as early as 1976 [24]~\cite{Ri76}) and efficiently applied in many contexts.

However, we note that this heuristic is not robust. In fact, it can have a very bad performance in the worst-case. A simple example which forces this worst-case is the following: consider a dictionary of $d$ binary strings (alphabet $\{0,1\}$) containing all the strings of length $\log d$ bits and consider searching for a string $x$ with error threshold $1$. If we cut the query string into two pieces of length $\log d/2$, we will know that at least $\Omega(\sqrt{d})$ strings will contain each of the two pieces. Thus at least $\Omega(\sqrt{d})$ dictionary elements will need to be considered. 
This artificial example shows that even with very short strings and a small number of errors, the query can have a very bad performance. 
Although the occurrence of such an extreme case is unlikely to happen on practical datasets, we will experimentally 
show that both the worst-case and average-case number of strings to be considered by a single query can considerably increase when the size of the data sets is increased. 

%When the strings are longer, the situation (both in practice and in theory) can be much better but the worst case can still be very bad. For example consider a dictionary where all the strings are of equal length $n/d$ and are all prefixed by the same prefix $p$ of length $n/2d$. Then if at query time we are also given a string prefixed by the same prefix, we will have to traverse all the $d$ elements of the dictionary. 

In this paper we present solutions that try to avoid the use of heuristics as much as possible. 
Our solution is based on the theoretical data structure presented in~\cite{Be09}. The solution presented there provably uses $O(n\log\sigma)$ bits of space, which is optimal up to a constant factor considering that the dictionary in its raw form occupies $n\log\sigma$ bits. The solution has a worst-case query time of $O(m+occ)$, where $occ$ is the number of reported occurrences and an expected $O(n)$ construction time. 

Our engineered variant also uses $O(n\log\sigma)$ bits, but its query and update times do not exactly match the ones of the theoretical solution. We will show however that the performance of our solution will be close to the previous one under assumptions that are likely to hold in practice. We verify experimentally that our solution is indeed competitive with previously proposed solutions and that the assumptions are realistic (they almost hold in practice). 

Our implementation can be downloaded at \url{http://code.google.com/p/compact-approximate-string-dictionary/}. It is licensed under the GNU Lesser General Public License (LGPL). 

In this paper we will state very few time or space bounds. Our time bounds will hold in a standard RAM model with word length $w=\Omega(\log n)$ where all standard arithmetic and logic instructions are assumed to run in constant time. For accuracy we will express our space bounds in terms of bits rather than in words.

\section{Related work}
The problem of approximate string dictionaries has received considerable amount of attention. We only review here some recent work. For a more thorough review of the existing literature we refer the reader to~\cite{Na01,Bo11}. 
Belazzougui~\cite{Be09} recently presented a dictionary specifically tailored for one edit error. The dictionary has a worst-case query time $O(m+occ)$ (where $occ$ is the number of reported occurrences), and uses asymptotically optimal $O(n\log\sigma)$ bits of space. Belazzougui and Venturini~\cite{BV12} presented compressed variants. In one of the variants, the dictionary uses $2nH_k+O(n)$ bits of space (where $H_k$ stands for the $k$th order entropy of the strings stored in the dictionary, with $k=o(\log_\sigma m)$) and answers queries in almost $O(m+occ)$ time. 

Our algorithm is in a sense similar to the FASTSS algorithm introduced in~\cite{BHHS08}. In this algorithm, a hash table stores all the strings that can be obtained by deleting (at most) $k$ characters from any string in the dictionary. Then at query time the hash table is queried for all the strings that can be obtained by deleting (at most) $k$ characters from the string. Then, it can be proved that all the strings at edit distance at most $k$ will be contained in the resulting set and the remaining task will be to test all the candidates. 
A clever improvement to FASTSS was achieved in~\cite{KLS10}, where instead of storing all the strings obtained by deleting one character, the dictionary stores pointers to the strings. Following this,~\cite{Bo12} compresses the strings in a different way: instead of storing the strings, only deleted characters along with their positions are stored. Space-efficient perfect hashing methods~\cite{cmph} are then used to store the lists of positions and characters. 

For more than one error, the work of Cole et al.~\cite{CGL04}, was the first to present good worst-case performance for any $k$. However the bounds are not competitive with the recent works both in terms of space and time for the case of $k=1$. Also the data structure  presented there seems rather complicated and incurs a polylogarithmic additive term in the query time which could in practice dominate the total query time. We are not aware of any implementation of that data structure.

\section{The data structure}
Our algorithm is based on the algorithm described in~\cite{Be09}, but which we heavily modify in order to get good practical performance. 
In order to make the algorithm practical we do the following:
\begin{enumerate}
\item We replace perfect hashing with more efficient linear probing. 
\item We mitigate the potential performance degradation of linear probing versus perfect hashing by using additional signature bits to make the query time even faster. 
\item We reduce the space used by the hash tables by using a compaction scheme that eliminates all the empty positions in the hash table. 
\item We suppress the two tries and the right-left signatures. Instead we directly compare the two strings using straightforward string comparison.  
\end{enumerate}
Moreover, we extend the algorithm to handle two or more errors. In the next subsections, we will give more details on our data structure. 

\subsection{Linear probing based dictionaries}
Our dictionary is composed of two basic components:
\begin{enumerate}
\item The substitution dictionary.
\item The exact dictionary.
\end{enumerate}
The substitution dictionary (or substitution store) is used for one purpose. Given a string $x\phi y$ with one wild card $\phi$~\footnote{A wild card is a special character that matches any character in the alphabet.}, report the set $S_{x\phi y}$ of all the characters which when substituted at that position will lead to a string in the dictionary. Storing exactly the substitution dictionary would take too much space. We instead use an approximate substitution dictionary that returns for each string $x\phi y$ a superset of the set $S_{x\phi y}$. Each returned character will then give us a candidate string which we need to check using the exact dictionary. 
The substitution store as originally proposed in~\cite{Be09} is composed of substitution lists. At construction time, the original algorithm does the following: for each word $x_i[1..n_i]$ of length $n_i$ and for each $j\in[1..n_i]$, append the character $x_i[j]$ at the end of the substitution list that corresponds to the string $x_i[1..j-1]\phi x_i[j+1..m_i]$. 

The exact dictionary will be used to locate the exact occurrences and will permit to check as to whether the candidate string is in the dictionary. 

Checking whether a word of length $m$ exists in a dictionary takes $O(m)$ time if implemented naively. Using bit-parallelism the checking can be improved to $O(\lceil\frac{m\log\sigma}{w}\rceil)$ time. 
In our case, we need to check the existence of a word obtained by modifying the original word via an insertion, a deletion or a substitution, where the characters to be substituted or inserted are obtained from the substitution store. We can easily preprocess the original word into a number of arrays such that computing any word obtained by applying an insertion, deletion or substitution can be done in constant time. Then checking that any modified word exists in the dictionary can also be done in $O(\lfloor\frac{m\log\sigma}{w}\rfloor)$ time by direct comparisons. 
In~\cite{Be09} a sophisticated scheme based on tries and on embedding signatures into the dictionary strings is used to reduce the time to check the existence of a modified word to $O(1)$. The sophisticated theoretical scheme has a significant space and time overhead which in practice will probably make it not competitive with a naive scheme that just compares the two words in a naive way. The reasons are two-fold. First, $m$ is usually small, and second the naive comparing will probably not induce more than one cache-miss as the compared words will probably fit in a one cache-line. 

We will use linear probing~\cite{Knuth63noteson} to implement both dictionaries. 
The reason why we use linear probing is that it has been proved that linear probing is among the most practical hashing methods. The good performance is essentially due to the good cache behavior~\cite{HL05}. In practice only one cache-miss is potentially incurred by any query to the hash table. This is especially important when the memory footprint of the data structure becomes very large and the cost of cache and TLB misses~\footnote{TLB stands for translation look ahead buffer which is a small in-CPU cache that serves the purpose of translating virtual to physical memory page addresses. When the number of accessed memory pages becomes too large, then there will be frequent TLB misses incurring additional slow memory accesses to refill the buffer entries.} tend to dominate the query time. 

In contrast, if we had chosen to use minimal perfect hashing to implement the exact dictionary, then a probe to the dictionary would have incurred between $3$ and $4$ memory probes (and correspondingly $3$ to $4$ potential cache misses) depending on the implementation. Perfect hashing implementations would incur between $2$ and $3$ memory probes and then retrieving the key from its place would have incurred further $1$ memory probe. In addition the original version of the substitution list which relied on minimal perfect hashing needs one additional component. This component is a prefix-sum data structure that stores the size of each substitution list. It is essentially an indexed bit vector of $2n$ bits. Thus a substitution list would have incurred between $4$ and $5$ memory probes. 

We note that another important reason for not choosing minimal perfect hashing is that it is by essence non-incremental and slow  to build. It also uses $O(n\log n)$ bits of space to build. 
%Furthermore it adds some complexity to the data structure by introducing a component that is not known to most practitioners. 

Finally we note that the main advantage of using perfect hashing comes from memory saving. We will show that using bit-based compaction, we can gain even more space (see Section~\ref{sec:bit_compact}).  

We now describe in more detail the linear probing scheme we use. Suppose that we use a hash table of size $t$, where the hash table slots are numbered from $0$ to $t-1$. We need to mark all the hash table elements as empty (for that we use a sentinel). Then, inserting an element $x$ into a hash table of size $t$ involves first computing $h(x)$ and then checking if position $i=h(x)\bmod t$ is empty. If this is the case then $x$ is inserted at position $i$ and the position $i$ is marked. Otherwise we check at position $i+1,i+2\ldots$, until we reach an empty position (if one of the position equals $t-1$, then the next one will be position $0$) and mark it. 
Queries are done as follows: if an element is mapped to location $i$, then we scan all consecutive locations starting from position $i$ until we find the element we are looking for or we reach an empty non-marked location. 

In our case, we use several hash tables to implement the exact dictionary and only one to implement the substitution store. For the former, the inserted elements are words of the dictionary (strings). 
For the latter, the inserted elements are either characters or pairs of characters and hash signatures. 
The insertion of a dictionary word (string) of length $m$ in the exact dictionary takes expected time $O(m)$ and searching for a word in the exact dictionary takes expected $O(m)$ time as well. 

We note that the number of distinct word lengths in the dictionary can not exceed $\sqrt{n}$. This is because the total number 
of strings of length more than $\sqrt{n}$ can not exceed $\sqrt{n}$. 
The exact dictionary can then be implemented by partitioning the words into at most $\sqrt{n}$ groups of equal-length strings, and use a separate hash table for each. 

The substitution store hash table is implemented as follows. For each word $x_i[1..n_i]$ of length $n_i$ and for each $j\in[1..n_i]$, insert the character $x_i[j]$ in the substitution hash table (using linear probing) at the position $i=h(x_{i,j})\bmod t$, where $x_{i,j}=x_i[1..j-1]\phi x_i[j+1..m_i]$.

For analyzing the performance of insertions in the substitution store, we first note that all the characters that belong to the same substitution list will be mapped to the same location in the hash table. That is, two characters $c_1$ and $c_2$ will be mapped to the same position $i=h(p\phi v)$ in the hash table whenever there exists a pair of strings $(p,v)$ such that both $pc_1v$ and $pc_2v$ belong to the dictionary. 

Thus the insertions of those characters will collide and slow down the insertion. More importantly, those collisions can considerably slow down querying. Consider that in the worst case a substitution list could have up to $\sigma$ elements in it which will all collide on the same position of the substitution hash table. Now if another substitution list with just one element is mapped to the same position, then a query to both substitution lists will have to traverse more than $\sigma$ elements even though the second substitution list has only one element in it. However we will argue that in practice it is very unlikely that many substitution lists will have more than one element in them. Indeed in the experiments section, we will see that the performance of our data structure is very stable and does not depend on this number. 
 
\subsection{Hash signature}
In order to further reduce the query time, we will append $r=O(1)$ bits to each entry of the substitution store. 
More precisely if we assume that a constant number of slots are explored, then using some bits can ensure that less than $1$ slot on average will generate a candidate string. The hash signature is computed based on the hash value of the string $x_i[1..j-1]\phi x_i[j+1..m_i]$, and is only $r$ bits long. When inserting one character that belongs to a substitution list, we will append the signature of its corresponding substitution list. The exact role of the signature is to allow filtering of the hash table characters at query time. We will only consider that a character is part of a substitution list if its corresponding signature is the same as that of the substitution list. Thus the signatures will help mitigate the effect of collisions between substitution lists.

\subsection{Bit-array compression}
\label{sec:bit_compact}
We will make use of a bit-array based compression similar to Google sparsehash~\cite{GSH}~\footnote{There are two main differences between sparsehash and our scheme. The first one is that sparsehash relies on quadratic probing rather than linear probing. The other one is that sparseshash uses dynamic rather than static memory allocation since it needs to divide the slots into blocks allocated using dynamic memory allocation in order to allow updates to the table.}. One drawback of this method is that it makes our dictionary no longer incremental. Given a hash table with load factor $\alpha$ (a hash table containing $n/\alpha$ slots) the compaction outputs two vectors: a compacted hash table of $n$ slots and a bit vector of $n/\alpha$ bits. It works in the following way: we traverse all the hash table slots and write a bit-vector where we write a zero for each empty hash table slot and a one for each non-empty hash table slot. During the traversal, we append every non-empty slot to the compacted hash table initially empty. 

In order to support queries on the compacted hash table, we will index the bit vector so as to support rank operations~\cite{jacobson1989space,munro1996tables} (a rank for position $i$ operation over a bit vector $b$ counts the number of ones in $b[1..i]$ ). We use a very simple rank implementation that increases the original space of the bit vector by a factor $1+1/\delta$ for some small constant $\delta$. We sample every $\delta$ word of the bit vector and store the rank up to that position. This ensures space $N(1+1/\delta)$ for a bit vector that has $N$ bits. 

Doing a rank query on the bit-vector will return a position in the compacted hash table and then we count the number of consecutive ones in the bit-vector to get the number of non-empty slots. 
\subsection{Checking Occurrences}
To check for potential occurrences, we use a strategy similar to Boytsov~\cite{Bo12}. Instead of computing the edit distance between the query string and the candidate string, we directly construct a modified query string by applying the candidate edit operation and then do a direct comparison between the obtained string and the candidate string. It can easily be shown that obtaining the modified query string can be done in amortized constant time. For example, a candidate string $p'$ for insertion of a character $c_1$ after position $i+1$ can be obtained from a candidate string for insertion of a character $c_0$ after position $i$ by turning back the character number $i+1$ to be equal to $p[i+1]$ and modifying character $i+1$ to $c_1$. That is initially we had $p[1..i]c_0p[i+1..m]$ and we get $p[1..i+1]c_1p[i+2..m]$ in constant time by changing just two characters. 

Then comparing two strings can take advantage of the strong processor word-parallelism to compare two candidate strings in time $O(\lfloor\frac{m\log\sigma}{w}\rfloor)$ time. 
As it was argued by Boytsov, it is expected that such a strategy will run much faster than the one based on computing the edit  distance between the candidate string and the modified query string. 

\subsection{Summary}
We can now state two theorems that summarize the performance of our data structure, based on the two assumptions formalized by the following definitions. 
\begin{definition}
A string hash function is said to be incremental if after some preprocessing done on an input string $x$, computing the hash value of any string at distance one from $x$ takes constant time. 
\end{definition}
\begin{definition}
A hash function that maps elements from a set $U$ into a set $V$ is said to be fully random, if it behaves as if it was randomly chosen from the set of all possible functions from $U$ into $V$. 
\end{definition}
The construction space and final space usage of our data structure is summarized by the following theorem. 
\begin{theorem}
\label{theo:space_one_error}
The dictionary presented in this section occupies $O(n\log\sigma)$ bits (where $n$ is the total length of the strings in the dictionary). The peak space usage during the construction is $O(n\log\sigma)$. If the bit-compaction is used, then the final dictionary occupies at most $n(2\log\sigma+O(1))$ bits of space. 
\end{theorem}
\begin{proof}
The words of the dictionary given as input occupy $n\log\sigma $bits of space. The substitution store has $n/\alpha=O(n)$ slots of length $\log\sigma+r$ bits each out of which only $n$ are occupied. The exact dictionary has $d$ slots occupying a total of $(n/\alpha)\log\sigma=O(n\log\sigma)$ bits of space. The peak space of the construction algorithm is clearly $O(n\log\sigma)$ bits, since it only reads the input and writes the output which occupy a total $O(n\log\sigma)$ bits of space. 
When compaction is used, then the space for the exact dictionary is reduced to $n\log\sigma+d(1+1/\alpha)+o(d)=n\log\sigma+O(d)$ bits and the space for the substitution store is reduced to $n(1+r+\log\sigma)+o(n)=n\log\sigma+O(n)$ bits 
for a total of $n(2\log\sigma+O(1))$ bits of space. 
\end{proof}
The performance of the construction and query time of the data structure is summarized by the following theorem. 
\begin{theorem}
\label{theo:time_one_error}
Assuming fully random and incremental hashing and that each substitution list contains a single element, then given a query string of length $m$, the expected query time is $O(m\lceil\frac{m\log\sigma}{w}\rceil)$ and given dictionary whose total string length is $n$, the expected construction time is $O(n)$. 
\end{theorem}
\begin{proof}
Assuming incremental hashing and given a string of length $m$, computing the hash value for querying or inserting in each substitution list takes amortized $O(1)$ time, resulting in total $O(m)$ time for all the $O(m)$ substitution lists accessed by a query or an insertion. The assumption that each substitution list is of size $1$ means that the $n$ strings $x_{i,j}=x_i[1..j-1]\phi x_i[j+1..m_i]$ used as keys for insertions in the substitution hash table will all be distinct. Furthermore, it is well known the cost of inserting an element in a linear probing hash table takes expected $O(1)$ time if one assumes fully random hashing~\cite{Knuth63noteson}. Then it is obvious that the insertion of every substitution list element takes constant expected time giving a total $O(n)$ time to insert all the $n$ elements in the substitution lists for all $d$ words in the dictionary. It follows that the total time taken by the construction of the non-compacted version of the substitution store is $O(n)$ 
~\footnote{To see why the assumption that every substitution list be of size $1$ is useful to ensure expected constant time insertion for substitution list elements, one could consider substitution lists of maximal size $\sigma$. Since all elements of such a substitution list will be mapped to the same location in the hash table, the insertion of every element will have to traverse all the elements already inserted in the same substitution before reaching an empty slot, resulting in expected $\Omega(\sigma)$ insertion time.}
By similar arguments, the construction time of the exact dictionary will be expected $O(n)$ time, since inserting a string of length $n_i$ involves expected constant number of accesses to the hash table and each access cost $O(n_i)$ time. Summed up over all the $d$ words in the dictionary, the total insertion time becomes $O(n)$. Generating the compacted version of the substitution store from the non-compacted one takes trivially $O(n)$. We thus have that the construction time of both the compacted and non-compacted versions of the dictionary presented in this section is $O(n)$. 

We now prove that the expected time of a query for a string of length $m$ is $O(m\lceil\frac{m\log\sigma}{w}\rceil)$. 
Recall that retrieving the elements of a substitution list needs to probe successive hash table slots starting from a position $i$ until an empty slot is encountered and then collect all the characters from the traversed slots. We will distinguish two cases. The case where the substitution list is non-existent, in which case the initial position $i$ will follow a uniform distribution over all the slots in the hash table. Then the total number of traversed slots will follow the bound of an unsuccessful search in a linear probing hash table which is known to be constant in expectation~\cite{Knuth63noteson}. 
In the case of a query for an existing substitution list, then the number of traversed slots is also expected constant. 
First note that the total number of traversed slots from position $i$ is exactly the same as that needed for an insertion or unsuccessful search in the table starting from this position $i$. Note also that the slot number $i$ is non-empty. 
Consider now that the hash table is decomposed into runs of full slots followed and runs of empty slots. Since the empty slots are constant fraction of the total number of slots, we conclude that choosing a starting position uniformly at random from the non-empty slots will result also in an expected constant number of traversed slots before reaching an empty slot. Since every non-empty position represents a substitution list, we conclude that the expected number of candidate characters collected when querying for an existing substitution list is also constant. This implies that the total number of candidate strings generated from querying all substitution lists is expected $O(m)$. When querying the exact dictionary for a candidate string, it takes expected $O(\lceil\frac{m\log\sigma}{w}\rceil)$ time to compare the candidate string with the strings from the dictionary. This is because the exact dictionary is represented using linear probing and a query will traverse an expected constant number of slots, each time spending $O(\lceil\frac{m\log\sigma}{w}\rceil)$ time to compare the string from the dictionary with the candidate string. This finishes the proof that the running time of a query is expected $O(m\lceil\frac{m\log\sigma}{w}\rceil)$. 

%Since choosing a slot uniformly at random from a run of non-empty slot will traverse about half of the elements of the slot, we conclude that the length of the run is at most twice the expected cost of an insertion that starts at a position chosen uniformly at random in the run. Thus an insertion that starts anywhere in the run will traverse at most twice the number of slots traversed by an insertion that starts at a position chosen uniformly at random inside the slot. We thus conclude that the expected 
\end{proof}
In practice, we expect that most of the substitution lists contain a single element (we precisely measure that this is actually true on the tested datasets). The full randomness assumption could have been removed, had we used tabulated hashing~\cite{PT12}. 
However as is done in most practical papers, we prefer to use faster hash functions that behave well in practice instead of trying to use slower hash functions with better theoretical guarantees~\footnote{The reader can refer to~\cite{MV08} for a theoretical justification of why simple hash function behave in practice as if they were fully random}.

\subsection{Extension to two or more errors}

We can show that our algorithm can be extended to work with two or more errors. %We have implemented the algorithm for two errors.
%The code can be downloaded at (give address of our code) and is licensed under the GNU Lesser General Public License. 
For two errors we store two substitution stores. The first substitution store is used for one error (the level-1 substitution store). The second substitution store (the level-2 substitution store) will store for every word and every pair of distinct positions, the character at the first position. 
At query time, we first query the level-2 substitution store, then, for each character retrieved from the substitution list, substitute the first wild card with the retrieved character and go to level-1 substitution store in order to find all the characters to be substituted at the second wild card. We then substitute those characters and finally query the exact dictionary for the resulting strings. 

For example, for the word $\texttt{ALABAMA}$, we store a substitution list associated with $\texttt{A}\phi\texttt{ABA}\phi\texttt{A}$. We first query the substitution list to get a set of characters that could be substituted to the first wild card. Then for each such character $c$, we query the level 1 for substitution list with $\texttt{AcABA}\phi\texttt{A}$. In our example, we have $c=L$ and thus we will query the substitution list $\texttt{ALABA}\phi\texttt{A}$ which contains a single element consisting in the character $M$.

For checking the occurrences we use the same strategy as before. Generating the candidate occurrences takes $O(1)$ time and checking time $O(\lfloor\frac{m\log\sigma}{w}\rfloor)$. 

Concerning the space usage of the scheme, we note that the word $\texttt{ALABAMA}$ being of length $\ell=6$, will contribute $\frac{\ell(\ell-1)}{2}=\frac{6\cdot 5}{2}=15$ characters to the substitution lists. 

In the general case, we can state the following theorems which summarize the performance of our data structures for $2$ errors. 
We first summarize the space-usage. 
\begin{theorem}
The dictionary presented in this section occupies $O(N\log\sigma)$ bits, where:
$$N=\sum_{i=1}^{d}{n_i(n_i-1)/2}$$
where $n_i$ is the length of the word number $i$ of the dictionary. The peak space usage during the construction is $O(N\log\sigma)$ bits of space. If the bit-compaction is used, then the final dictionary occupies at most $(N+2n)(\log\sigma+O(1))$ bits of space. 
\end{theorem}
\begin{proof}
The bounds can be proved using the same arguments used for theorem~\ref{theo:space_one_error}. The only difference 
is that that we now have in addition level-2 substitution store in which we insert $N$ elements. The space bounds easily follow. 
\end{proof}
And then analyze the query time. 

\begin{theorem}
Assuming fully random and incremental hashing and that each substitution list contains a single element, then given a query string of length $m$, the expected query time is $O(m^2\cdot \lceil\frac{m\log\sigma}{w}\rceil)$.
Given a dictionary of $d$ words such that $$N=\sum_{i=1}^{d}{n_i(n_i-1)/2}$$
where $n_i$ is the length of the word number $i$ of the dictionary. Then the expected construction time of the data structure is $O(N)$.  

\end{theorem}
\begin{proof}
Similarly to the space bounds, the time bounds can be also proved using the same arguments used for proving theorem~\ref{theo:time_one_error}. The cost of queries and insertions in the level-2 substitution store can be bounded in the same way as that of level-1 substitution store, except that we replace the term $m$ by $m^2$ for the queries since every query will probe in addition $O(m^2)$ level-2 substitution lists. We also replace the term $n$ by $N$ in the construction time, since we will have to insert to $N$ elements in the level-2 substitution store. 
\end{proof}

We have only implemented the algorithm for two errors, but our dictionary can be extended to handle $K>2$ errors. For that we use $K$ levels of substitution stores, where a word of length $n_i$ contributes ${{n_i}\choose{k}}\cdot \log\sigma$ bits to a level-$k$ substitution store ($1\leq k\leq K$).

\subsection{Implementation details}
Our implementation is modular. It consists of a dictionary for zero errors. If we want to support one error, then we add a 1-level substitution store. If we want to support two errors, then we add another substitution store on top of the first two components. 

All hash tables are implemented using linear probing with the same predefined load factor $\alpha<1$. To implement the exact dictionary, we store words of different lengths using different hash tables. Instead of storing $\sqrt{d}$ different hash tables 
for the $\sqrt{d}$ distinct word lengths, We use a threshold $\beta$ and use $\beta$ different hash table. All words of length $i<\beta$ will be stored in hash table number $i$ which consists in $\lceil n_i/\alpha\rceil$ slots where each slot is of length $i$ characters. All words of length $i\geq \beta$ will be stored in a hash table that stores pointers to the original words instead of the words themselves. 
In our implementations, we set $\beta=16$ characters. 

We now describe our implementation of the substitution stores, which will also be implemented using the hash tables based on linear probing. As described before, we have two variants on how to implement the hash tables, one in which each entry is just a character and the other where each entry is a character and a signature. We encode a character using one byte ($8$-bits) and encode a signature using $4$ bits. The encoding of signatures with $4$ bits simplifies the implementation. The hash table is divided into blocks of two slots occupying $3$ bytes each, the first byte is the character of the first slot, the third byte is the character of the second slot and the central (second byte) stores the signatures of the two slots. 

A query for a substitution list may have to traverse many slots in the hash table before reaching an empty one. 
In order to bound the worst-case query time, we have implemented the following strategy. Whenever a query for a substitution list 
traverses more than $\sigma$ slots in the hash table, we stop traversing the hash table, and just create $\sigma$ 
candidate strings from all possible $\sigma$ characters~\footnote{Note that traversing $\sigma$ slots does not 
mean that we will have collected $\sigma$ distinct candidate characters. The number could for example, be much smaller 
when hash signatures are used.}.

All our linear probing hash table are parameterized with a load factor $\alpha<1$. For simplicity we use the same parameter for all hash tables. If a hash table holds $n$ keys then its capacity will be $\lceil n/\alpha\rceil$. 

As for the hash function, we use a polynomial hash function (Rabin Karp hash functions~\cite{KR87}) modulo $2^{32}-5$. The hash values are computed modulo $2^{32}-5$, the largest prime that is smaller than $2^{32}$. The hash functions are very simple, using a randomly chosen seed $r\in[1..2^{32}-6]$, the hash value for a string $x$ is computed by $h(x)=\sum_{i=1}^{m}x_m\cdot r^m$. The polynomial hash function has a nice property of being incremental. We can in time $O(m)$ preprocess the string $x$, computing a certain number of vectors of $O(m)$ integers such that computing the hash value of any string at distance one from $x$ takes constant time (using a constant number of additions and multiplications modulo). See the details in~\cite{Be09}.

In case we use compaction, we will implement a bit vector of size $n'=\lceil n/\alpha\rceil$ bits that marks all non-empty slots and then build a compact hash table that stores only the non-empty $n$ slots. In order to support rank operations on blocks, we use a very simple implementation where we store the partial counts of the number of ones every $\delta$ words (we assume $32$-bit words). The partial count number $i$ will store the number of ones in the first $(i-1)\delta$ words. Instead of storing the partial counts in a separate vector, we instead store them interleaved with the bit vector (a partial count followed by $32\delta$ bits of the original bit vector). In order to support efficient counting of ones up to position $i$, we use the bit vector to get the count of the number of ones up to word number $\lfloor i/\delta\rfloor$ using the partial counts and then count the number of ones in the up to $32\delta-1$ bits which will span at most $\delta$ words using the popcount operation~\cite{bithacks} (although the instruction is not constant time, in practice it is very fast and can be considered as constant time). Thus the total space will be $n'+n'/\delta=\lceil n/\alpha\rceil(1+\delta)$ bits and the time will be $O(\delta)$.

We note that all our non-compacted methods are fully incremental, in the sense that adding an element to the dictionary should be possible and roughly take the same time than the total building time divided by the total number of strings. In contrast, the compacted methods will not be incremental, because of the compaction steps. 

Our implementation assumes $32$ bit integers and pointers, but it is trivial to change the code to work with $64$ bits (and hence manage larger dictionaries) without incurring much space or time overhead. The only component that would use larger space would be the exact dictionary for large words which uses pointers. All other components do not use pointers at all, and thus their space usage should not be affected. 

\section{Experiments}
We experimented with the following 2 datasets: the Wikipedia dataset which has around 1.8 million words and an English dictionary which has about 213 thousand words~\footnote{We thank Dennis Luxen for providing us with the Benchmarks.}. These two data sets are also used in the experiments in Karch et al's paper~\cite{KLS10}. They also use two other smaller datasets, Mobydick and Town. 
We experimented with the Mobydick, but given the small size of the datasets (37 thousand words), the result did not give any new insight (the whole data structure would fit in the CPU cache). We did not do experiments on Town as it was not available and anyway, the result would not have been much different from Mobydick given that the file contains only 47 thousand words. 

We implemented the data structure, the construction and query algorithms in C language using GNU GCC Compiler version 4.4.1. The  tests were done on an Intel 3.0 Gigahertz core 2 duo e8400 processor running Windows 7. We only used a single core. 

We experimented with different load factors (parameter $\alpha$) ranging from $0.3$ to $0.7$. The results are summarized in tables~\ref{table:compar_english1}~,~\ref{table:compar_wiki1}~,~\ref{table:compar_english2}~,~\ref{table:compar_wiki2}.  In the tables, LF stands for load factor. In the bit vector implementations, we experimented with $\delta=4$. That is we put a count every $4$ words. The results for the compacted versions of our indices on the English dictionary dataset are shown in figures ~\ref{fig:plot_compact_english1} and ~\ref{fig:plot_compact_english2}.

We compared our results with those of Karch et al. The reason we have chosen to do so is that the results of Karch et al. have proved to be faster than the competing approaches.
As the implementation of Karch et al. is not publicly available, we just took the results from Karch et al. paper. We note that the hardware used by Karch et al. is comparable to ours although they are not identical~\footnote{Their experiments where done on a on an Intel Xeon X5550 CPU at 2.67 GHZ which is a machine with very close performance to ours.}. Thus the query time comparison is not be completely accurate. However, given the close capacity of the two machines we may consider the comparison to be accurate up to a margin of error between $10\%$ and $20\%$. 
On the other hand the space usage is directly comparable, as it does not depend on the used hardware. 

To carry out our tests, we randomly choose $1000$ words from our dictionary and apply one randomly chosen edit operation on each of them and then query the dictionary. We then divide the time by $1000$ to get the final query time. We repeat the procedure $20$ times and average the result. 

%We did not compare with Boytsov's result. Although the source code is available, we did not manage to make it run. 
%We could have compared with Boytsov~\cite{Bo12} results as stated in his paper but there are no tables indicating the performance of the algorithms in his paper. The only alternative would have been to guess the results performance from the figures which would have given inaccurate results. 
We did not compare with Boytsov's result. From the similarities with our method we expect Boytsov's method to have 
similar space bounds and query times. However the method is inferior to ours in the following aspects: 
\begin{enumerate}
\item It is very slow to build~\cite{Bo12}, much slower than the alternatives and ours is faster than that of Karch et al. which itself is more competitive. The construction time can go up to $1$ hour for large datasets. In contrast, our construction time is less than one minute for all datasets. 
\item The use of minimal perfect hashing leads to a very large temporary space at construction time, even though the final space is small. The minimal perfect hashing method usually require at least $12$ bytes per entry, even though the final space is less than $3$ bits. In contrast the construction space of our non-compacted variant is exactly the same as the final space. 
\item The method is not incremental. The insertion of new words requires rebuilding the whole data structure. 
\end{enumerate}
\begin{figure}
\begin{center}
\includegraphics[scale=0.4]{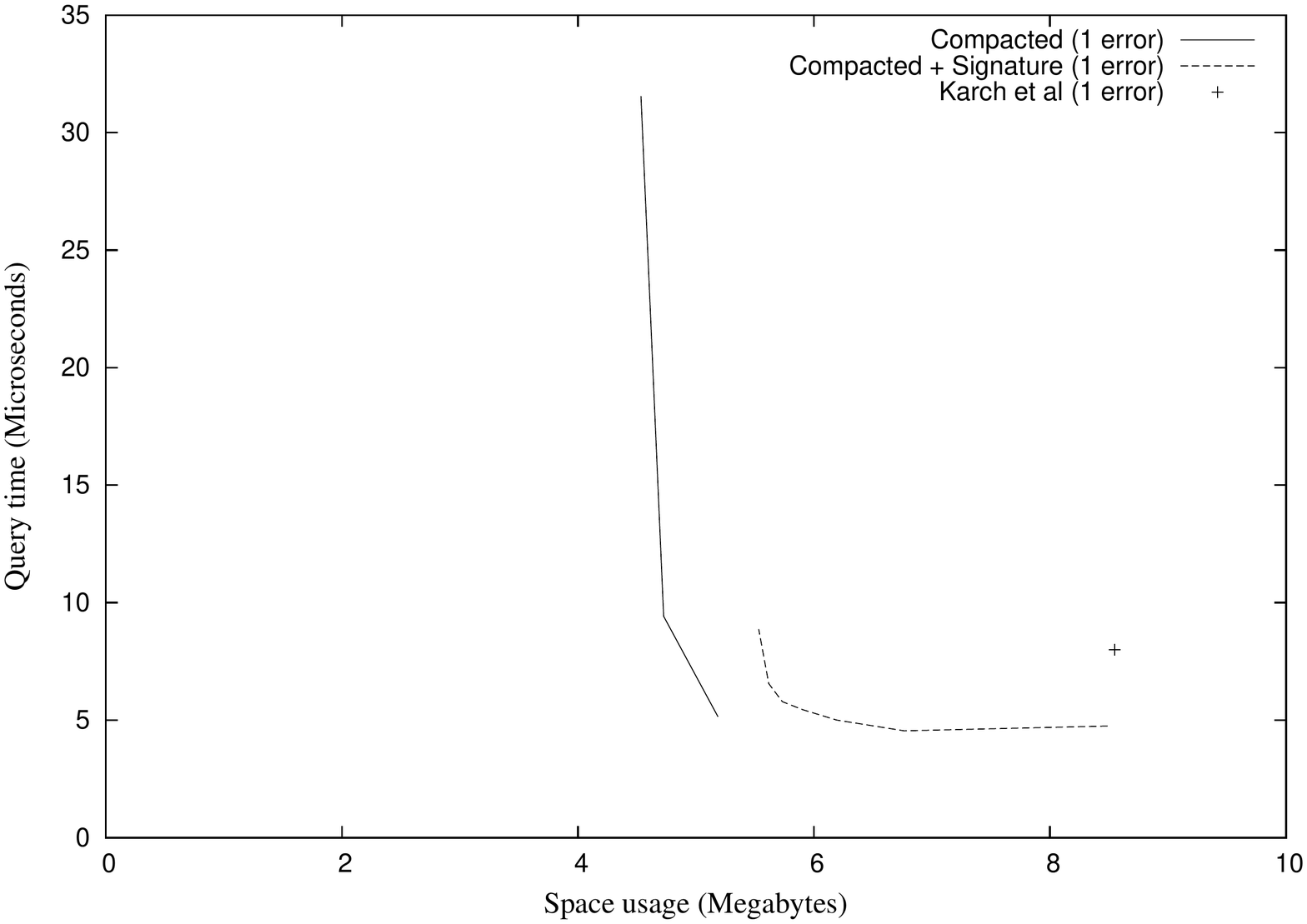}
\end{center}
\caption{Space usage versus query time on the English dictionary dataset with one error. Only compacted versions are shown, 
where we vary the load factor from 0.3 to 0.7. For Karch et al only one point is shown since the 
scheme for one error does not admit tradeoffs.}
\label{fig:plot_compact_english1}
\end{figure}

\begin{figure}
\begin{center}
\includegraphics[scale=0.4]{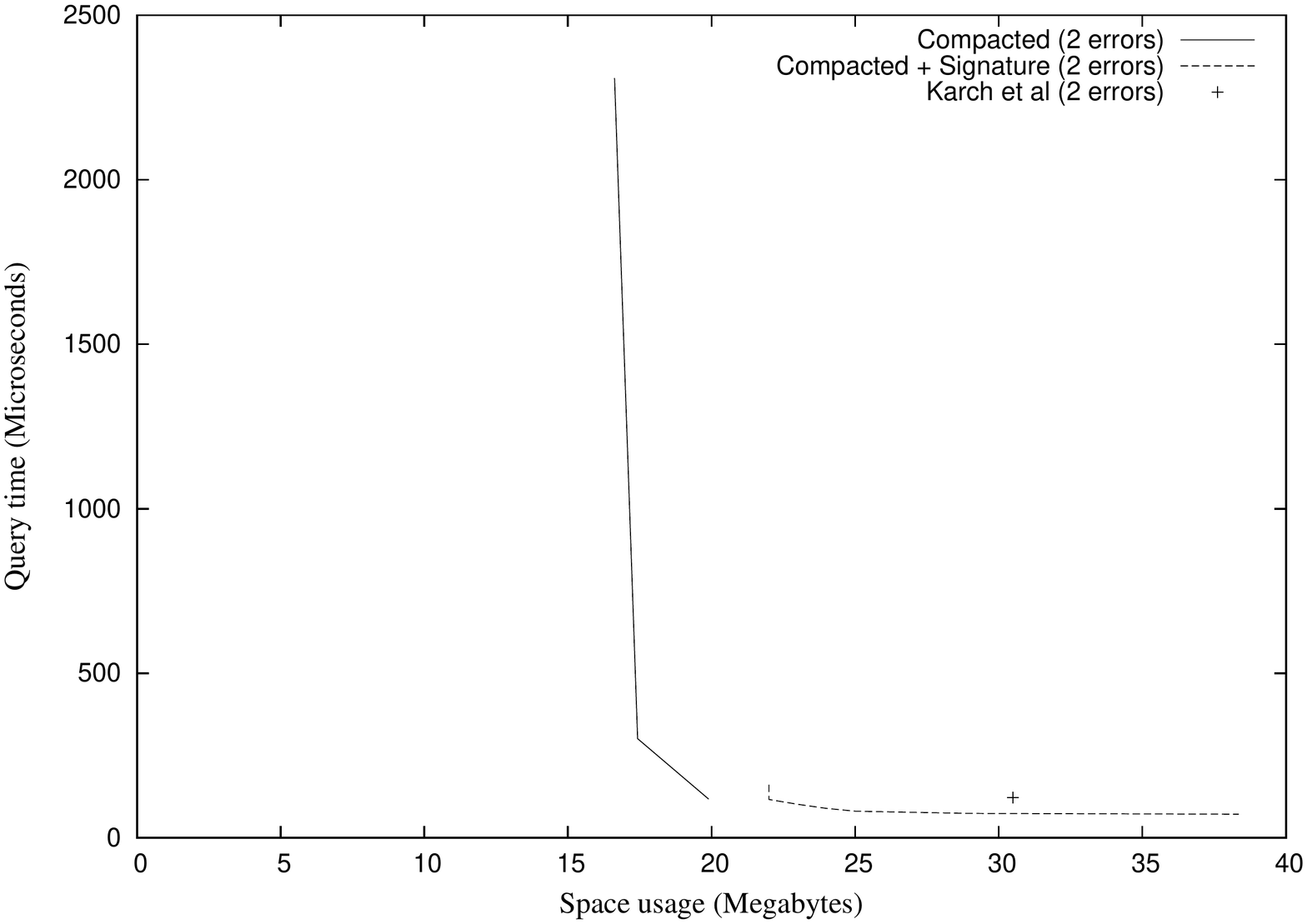}
\end{center}
\caption{Space usage versus query time on the English dictionary dataset with two errors. Only compacted versions are shown, where we vary the load factor from 0.3 to 0.7}
\label{fig:plot_compact_english2}
%\vspace{-6mm}
\end{figure}

\begin{table}
\small
\begin{center}
\begin{tabular}{|l|c|c|c|}
\hline
Method & Constr. time (seconds) & Space (Mb) & Query time ($\mu s$)\\
\hline
Karch et al     & 3.45
		& 8.55 
		& 8 \\

Non comp.	(LF 0.7) & 0.4
		& 5.72 
		& 27.6 \\
Non comp.  + Sign. (LF 0.7) 
		& 0.44 
	     	& 7.15
		& 6.24 \\
Compacted (LF 0.7)	& 0.43
		& 4.53
		& 31.6 \\

Compacted (LF 0.3)	& 0.47
		& 5.19
		& 5.154 \\

Comp. + Sign. (LF 0.7)
		& 0.48
		& 5.53
		& 8.9 \\
Comp. + Sign. (LF 0.3)
		& 0.625
		& 6.77
		& 4.55 \\
\hline

\end{tabular}

\end{center}
\caption{Comparison of existing methods on the English dictionary dataset for one edit error}
\label{table:compar_english1}
\end{table}

\begin{table}
\small
\begin{center}
\begin{tabular}{|l|c|c|c|}
\hline
Method & Constr. time (seconds) & Space (Mb) & Query time ($\mu s$)\\
\hline
Karch et al     & 32.29
		& 55.84
		& 34 \\
Non comp. + Sign. (LF 0.7) 
		& 4.61 
	     	& 54.44
		& 8.44 \\
Compacted (LF 0.5)	& 4.72
		& 36.27
		& 16.6 \\
Comp. + Sign. (LF 0.7)
		& 4.91
		& 42.3
		& 13.5 \\
Comp. + Sign. (LF 0.3)
		& 5.17
		& 47.5
		& 8.28 \\

\hline

\end{tabular}

\end{center}
\caption{Comparison of existing methods on the Wikipedia dataset for one edit error}
\label{table:compar_wiki1}
\end{table}

\begin{table}
\small
\begin{center}
\begin{tabular}{|l|c|c|c|}
\hline
Method & Constr. time (seconds) & Space (Mb) & Query time ($\mu s$)\\
\hline
Karch et al     & 12.596
		& 30.49 
		& 122 \\
Non comp. + Sign.	(LF 0.7) 
		& 2.612 
	     	& 27.87
		& 117 \\
Compacted (LF 0.7) & 2.54
		& 16.6
		& 2309\\

Comp. + Sign. (LF 0.7)
		& 2.802
		& 22.8
		& 162 \\
Comp. + Sign. (LF 0.4)
		& 2.95
		& 24.44
		& 89.27 \\
Comp. + Sign. (LF 0.2)
		& 3.1
		& 29.08
		& 73.89 \\
\hline
\end{tabular}

\end{center}
\caption{Comparison of existing methods on the English dictionary dataset for two errors}
\label{table:compar_english2}
\end{table}

\begin{table}
\small
\begin{center}
\begin{tabular}{|l|c|c|c|}
\hline
Method & Constr. time (seconds) & Space (Mb) & Query time ($\mu s$)\\
\hline
Karch et al     & 107.289
		& 170.79
		& 502 \\
Non comp. + Sign.	(LF 0.7) 
		& 21.29
	     	& 203.82
		& 503 \\

Compacted (LF 0.5)
		& 21.36
		& 127.87
		& 1993	\\
Comp. + Sign. (0.7)
		& 22.86
		& 162.57
		& 1271 \\
\hline

\end{tabular}
\end{center}
\caption{Comparison of existing methods on the Wikipedia dataset for two errors}
\label{table:compar_wiki2}
\end{table}

From the results, we can make the following observations:
\begin{enumerate}
\item All our algorithms are much faster to build than the reference algorithm. They are usually between $5$ and $10$ times faster to build. 
\item The algorithm for one edit error is very robust. It is up to $4$ times faster than Karch et al's result on very big datasets (Wikipedia). 
\item The compaction gives a very useful trade off between space and time. 
\item The query time of our method for two errors is not so good especially on the very large data sets. However it still provides a relevant space/time trade off. Even with the largest dataset, we can achieve very significant space advantage at the expense of a larger query time. 
\end{enumerate}

The non-compacted variants of our index naturally support insertions of new words. We did some experiments 
to back our claim that our index efficiently supports insertions of new elements. To that purpose, we first built our index 
on a fraction $0.7$ of the words in the Wikipedia dataset with a load factor of $0.7$. We then inserted a fraction $0.25$ of the 
original dataset (which means that the final load factor is $0.95$). We measured the insertion time~\footnote{We did the experiments on a different machine with slightly weaker performance (Intel Xeon with 2.27 GHz, 2 Gigabytes of RAM and Windows  xp operating system).} of the successive $0.05$ fractions of the original file. The results are reported in table~\ref{table:insert_test_wiki}. As can be seen from the table, the insertion takes only few microseconds for the one error dictionary 
and few dozens of microseconds on the two error dictionary.

\begin{table}
\small
{\centering
\begin{tabularx}{\textwidth}{|X|X|X|}

\hline

Load Factor & \multicolumn{2}{>{\centering\setlength{\hsize}{2\hsize}\addtolength{\hsize}{2\tabcolsep}}X|}
{Average insertion time in $\mu s$} \\

\cline{2-3}&1 error      &2 errors\\
\hline
0.7-0.75&2.42&13.34\\
0.75-0.80&2.44&13.83\\
0.80-0.85&2.52&16.66\\
0.85-0.90&3.52&30.11\\
0.90-0.95&5.56&43.74\\

\hline

\end{tabularx}
}
\caption{Insertion times in Wikipedia indexes for one and two errors}
\label{table:insert_test_wiki}

\end{table}

We experimentally verified the assumption that most of the substitution lists (for one edit error) contain one element. From table~\ref{table:subst_list_stats}, we can see that more than $96\%$ of the characters are in substitution lists of size $1$. 

\begin{table}[h]
\small
\begin{center}
\begin{tabular}{|c|c|c|}
\hline
List size & Percentage (English) &Percentage (Wikipedia)\\
\hline
1	        & 98.89\%
		& 96.16\%\\

2	        & 0.71\% 
		& 1.63\%\\

3		& 0.19\% 
		& 0.52\%\\

4		& 0.09\%
		& 0.29\%\\

5		& 0.05\%
		& 0.19\%\\

$\geq 6$	& 0.08\%
		& 1.20\%\\

\hline

\end{tabular}

\end{center}
\caption{Percentage of elements in substitution lists of given size for the two datasets}
\label{table:subst_list_stats}
\end{table}

We will now show that the heuristic based on dividing the dictionary and query strings into $k+1$ pieces, is not as robust 
as ours. Consider the following algorithm for the case $k=1$, where each query and dictionary string is divided into $2$ pieces. A given string $s_i$ of length $n_i$ is divided into a prefix $p_i$ of length $\lfloor n_i/2\rfloor$ and a suffix $q_i$ of length $\lceil n_i/2\rceil$. Then a pointer to the string $s_i$ is stored in the two lists of entries associated with the prefixes $p_i$ and $q_i$. A query for one substitution 
error on a string $x$ will then divide $x$ into two pieces, a prefix of length $\lfloor n_i/2\rfloor$ and a suffix of length $\lceil n_i/2\rceil$, and query for the lists of strings associated with the two pieces and check every candidate in the two lists. 
We simulated this process on our two datasets. We then queried the dictionary for every string $s_i$ in the dataset, and measured the average and maximum sizes of the candidate lists. The result is reported in table~\ref{table:heuristic_method_results}. 

\begin{table}[h]
\small
\begin{center}
\begin{tabular}{|c|c|c|}
\hline
Dataset & average list size & Maximal list size\\
\hline
English & 33.963307 & 613\\
\hline
Wikipedia & 232.028325 & 6923\\
\hline
\end{tabular}
\end{center}
\caption{Average and maximal size of candidate sets involved in queries used by the 
heuristic algorithm based on partitioning}

\label{table:heuristic_method_results}
\end{table}

As can be seen from the table, the maximal number of candidate strings for the larger dataset (Wikipedia)
is an order of magnitude larger than that for the smaller dataset (English). Even the average 
number of candidate strings is about $7$ times larger.

%In contrast as it can be shown the method of Karch et al suffers from a bigger slow down which is due to the larger number of candidates. This shows that the performance of these methods is not stable and can in fact be much worse for larger dictionaries. 

\section{Conclusion}
We have presented an efficient and robust algorithm for approximate string matching on dictionaries. 
%The general technique is a kind of pointer compression, which we believe can be used in other approximate string matching dictionaries. 
The algorithms we propose are easy to implement and very efficient. 
The only compromise we made was that we insert all elements in a substitution list in the same location of the substitution store. This could give bad worst-case performance as any substitution list can have up to $\sigma$ elements in it. However in practice we expect that most substitution lists will have only a single element in them (and this is indeed the case on our datasets). 

The algorithm for one error uses roughly 2 times the size of an exact dictionary. The query time is linear in expectation (the expectation is on the random numbers and not on the randomness of the datasets) based on some assumptions which are almost met in practice. 
The algorithm for two errors has super linear space use, and in practice is slower than the first one. However it still achieves a relevant time/space trade off. 

We also note that the non-compacted versions of our structures are incremental (with insertion time at worst few dozens of microseconds) which was not the case of the method of Boytsov~\cite{Bo12} (which besides has a very slow construction time). 

\section*{Acknowledgments}
We thank Nacera Bensaou for her encouragement and fruitful discussions at the early stages of this work. 
We thank Simon Puglisi for his many useful comments and remarks. The authors are grateful to the 
reviewers for their constructive comments and remarks. 
\bibliographystyle{plain}
\bibliography{approx_dict}

\end{document}